\documentclass[11pt]{article}
\usepackage{amsmath}
\usepackage{amsthm}
\usepackage{titlesec}
\usepackage{indentfirst}

\newtheorem{theorem}{Theorem}%[section]

\newtheorem{lemma}[theorem]{Lemma}%[section]

\newtheorem{definition}[theorem]{Definition}%[section]

\newtheorem{remark}[theorem]{Remark}%[section]

\begin{document}

\title{One-way finite automata with quantum and classical states\thanks{This work is supported in
part by the National Natural Science Foundation of China (Nos.
60873055, 61073054), the Natural Science Foundation of Guangdong
Province of China (No. 10251027501000004), the Fundamental
Research Funds for the Central Universities (Nos.
10lgzd12,11lgpy36), the Research Foundation for the Doctoral
Program of Higher School of Ministry of Education (Nos.
20100171110042, 20100171120051) of China,  the Czech Ministry of
Education (No. MSM0021622419), the China Postdoctoral Science
Foundation project (Nos. 20090460808, 201003375), and the project
of  SQIG at IT, funded by FCT and EU FEDER projects projects QSec
PTDC/EIA/67661/2006, AMDSC UTAustin/MAT/0057/2008, NoE Euro-NF,
and IT Project QuantTel.}}

\author{Shenggen Zheng$^{1,}$\thanks{{\it  E-mail
address:} zhengshenggen@gmail.com},\hskip
2mm Daowen Qiu$^{1,3,4,}$\thanks{Corresponding author. {\it E-mail address:}
issqdw@mail.sysu.edu.cn (D. Qiu)},
\hskip 2mm Lvzhou Li$^{1,}$
\thanks{{\it  E-mail address:} lilvzhou@gmail.com},\hskip 2mm Jozef Gruska$^{2},$\thanks{{\it  E-mail
address:} gruska@fi.muni.cz}\\
\small{{\it $^{1}$ Department of
Computer Science, Sun Yat-sen University, Guangzhou 510006,
  China }}\\
\small{{\it $^{2}$ Faculty of Informatics, Masaryk University, Brno, 602 00, Czech Republic }}\\
\small {{\it $^{3}$ SQIG--Instituto de Telecomunica\c{c}\~{o}es, Departamento de Matem\'{a}tica,}}\\
\small {{\it  Instituto Superior T\'{e}cnico, TULisbon, Av. Rovisco Pais
1049-001, Lisbon, Portugal}}\\
\small{{\it $^{4}$ The State Key Laboratory of Computer Science, Institute of Software,}}\\
\small{{ \it Chinese  Academy of Sciences, Beijing 100080, China}}
}

\date{ }
\maketitle \vskip 2mm \noindent {\bf Abstract}
\par
 In this paper, we introduce and explore a new model of {\it quantum finite automata} (QFA). Namely,
 {\it one-way finite automata with quantum and classical states} (1QCFA), a one way version of
 {\it two-way finite automata with quantum and classical states} (2QCFA) introduced by Ambainis and Watrous in 2002 \cite{AJ}. First, we prove that
  {\it one-way probabilistic finite automata} (1PFA) \cite{AP} and
  {\it one-way quantum finite automata with control language} (1QFACL) \cite{ACB}
  as well as several other models of QFA, can be simulated by 1QCFA. Afterwards, we
 explore several closure properties for the family of languages accepted
 by 1QCFA. Finally, the state complexity of
 1QCFA is explored and the main succinctness result is derived. Namely,  for any prime $m$ and any
 $\varepsilon_1>0$, there exists a language $L_{m}$ that cannot be recognized by any {\it measure-many one-way quantum finite automata}
 (MM-1QFA) \cite{Kon97} with bounded error $\frac{7}{9}+\epsilon_1$, and any 1PFA recognizing it has at last $m$ states,
 but $L_{m}$ can be recognized by a 1QCFA for any error bound $\epsilon>0$ with $\bf{O}(\log{m})$ quantum states and 12 classical states.

\par

\section{Introduction}

An important way to get a deeper insight into the power of various quantum
resources and features for information processing is to explore power of
various quantum
variations of the basic models of classical automata. Of a special interest
and importance is to do that for various quantum variations of classical
 finite automata because quantum resources are not cheap and quantum
operations are not easy to implement. Attempts to find out how much one
can do with very little of quantum resources and consequently with the
most simple quantum variations of classical finite automata are therefore
of particular interest. This paper
is an attempt to contribute to such line of research.

There are two basic approaches how to introduce quantum features to
classical models of finite automata. The first one is to consider quantum
variants of the classical {\it one-way (deterministic) finite automata}
(1FA or 1DFA) and the second one is to consider quantum variants of the
classical {\it two-way finite automata} (2FA or 2DFA). Already the very first
attempts to introduce such models, by Moore and Crutchfields \cite{Moore} and
Kondacs and Watrous \cite{Kon97} demonstrated that in spite of the fact that in the
classical case, 1FA
and 2FA have the same recognition power, this is not so for their quantum
variations. Moreover, already the first important model of {\it two-way
quantum finite automata} (2QFA), namely that introduced by Kondacs and Watrous,
demonstrated that very natural quantum variants of 2FA are much too
powerful - they can recognize even some non-context free languages and
are actually not really finite in a strong sense. It started to be
therefore of interest to introduce and explore some ``less quantum"
variations of 2FA and their power \cite{Amb06,Amb-F,AJ,AN,ANT,ACB,Bro,LiQ2,LiQ3,LiQ5,MP,Pas,QM,Qiu4,AY2,AY3,AY4}.

A very natural ``hybrid" quantum variations of 2FA, namely, {\it two-way
quantum automata with quantum and classical states}  (2QCFA) were
introduced by Ambainis and Watrous \cite{AJ}. Using this model they were able
to show in an elegant way that an addition of a single qubit to a
classical model can enormously increase power of automata. A 2QCFA is
essentially a classical 2FA augmented with a quantum memory of constant
size (for states in a fixed Hilbert space) that does not depend on the
size of the (classical) input. In spite of such a restriction, 2QCFA have
been shown to be more powerful than {\it two-way probabilistic finite automata}
(2PFA) \cite{AJ}.

Because of the simplicity, elegance and interesting properties of the
2QCFA model, as well as its natural character, it seems to be both useful and interesting  to
explore what such a new ``hybrid" approach will provide in case of one-way
finite automata and this we will do in this paper by introducing and
exploring 1QCFA.

In the first part of the paper, 1QCFA are introduced formally and it is
shown that they can be used to simulate a variety of other models of
finite automata. Namely, 1DFA, 1PFA, measure-once 1QFA (MO-1QFA) \cite{Kon97}, measure-many 1QFA (MM-1QFA) \cite{Kon97}
and {\it one-way quantum finite
automata with control language} (1QFACL) \cite{ACB}. Of a special interest
is the way how 1QCFA can simulate  1QFACL - an interesting  model the
behavior of which is, however,  quite special. Our simulation of 1QFACL
by 1QCFA allows to see behavior of 1QFACL in a quite transparent way. We
also explore several closure properties of the family of languages
accepted by 1QCFA.  Finally, we derive a result concerning the state
complexity of 1QCFA that also demonstrates a merit of this new model.
 Namely we show that for any prime $m$ and any
 $\varepsilon_1>0$, there exists a language $L_{m}$ than cannot be recognized by any MM-1QFA with bounded error $\frac{7}{9}+\epsilon_1$,
 and any 1PFA recognizing it has at last $m$ states, but $L_{m}$ can be recognized by a 1QCFA for any error bound $\epsilon>0$ with
 $\bf{O}(\log{m})$ quantum states and 12 classical states.

The rest of the paper is organized as follows. Definitions of all
automata models explored in the paper are presented in Section 2. In
Section 3 we show how several other models of finite automata can be
simulated by 1QCFA.  We also explore several closure properties of the
family of languages accepted by 1QCFA in Section 4. In Section 5 the
above mentioned succinctness result is proved and the last section
contains just few concluding remarks.

\section{Basic models of classical and quantum finite
 automata}

 In the first part of this section we formally introduce those basic models
of finite automata we will refer to in the rest of the paper and in the
second part of this  section, we formally introduce as a new model 1QCFA.
 Concerning the basics
of quantum computation we refer the reader to \cite{Gru99,Nie} and concerning
the basic properties of the automata models introduced in the following we
refer the reader to \cite{Gru99,Gru00,Hop,AP,LiQ4}.

\subsection{Basic models of classical and quantum finite
 automata}
In this subsection, we recall the definitions of DFA, 1PFA, MO-1QFA, MM-1QFA and 1QFACL.
\begin{definition}
A {\it deterministic finite automaton} (DFA) $\mathcal{A}$ is specified by a 5-tuple
\begin{equation}
\mathcal{A}=(S,\Sigma,\delta,s_0, S_{acc}),
\end{equation}
where:
\begin{itemize}
\item [1.]$S$ is a finite set of classical states;
\item [2.]$\Sigma$ is a finite set of input symbols;
\item [3.]$s_{0}\in S$ is the initial state of the machine;
\item [4.]$S_{acc}\subset S$ is the set of accepting states;
\item [5.]$\delta$ is the transition function:
\begin{equation}
\delta:S\times\Sigma \rightarrow S.
\end{equation}
\end{itemize}
\end{definition}

Let $w=\sigma_1\sigma_2\cdots\sigma_n$  be a string over the alphabet $\Sigma$. The automaton $\mathcal{A}$
accepts the string $w$ if a sequence of states, $r_0, r_1, \cdots, r_n$, exists in $S$ with the following conditions:
\begin{enumerate}
  \item $r_0=s_0$;
  \item $r_{i+1}=\delta(r_i, \sigma_{i+1})$, for $i=0, \cdots, n-1$;
  \item $r_n\in S_{acc}$.
\end{enumerate}
DFA recognize exactly the set of {\it regular languages} (RL).

\begin{definition}
A {\it one-way probabilistic finite automata} (1PFA) $\mathcal{A}$ is specified by a 5-tuple
\begin{equation}
\mathcal{A}=(S,\Sigma,\delta,s_{1},S_{acc}),
\end{equation}
where:
\begin{itemize}

\item[1.] $S=\{s_1,s_2,\cdots, s_{n}\}$ is a finite set of classical states;

\item[2.] $\Sigma$ is a finite set of input symbols; $\Sigma$ is then extended to the tape symbol
set $\Gamma=\Sigma\cup\{\ |\hspace{-1.5mm}c,\$\}$, where $\ |\hspace{-1.5mm}c\notin \Sigma $ is called the left end-marker and
$\$\notin \Sigma$ is called the right end-marker;

\item[3.] $s_{1}\in S$ is the initial state;

\item[4.] $S_{acc}\subset S$ is the set of
accepting states;

\item[5.] $\delta$ is the transition function:

\begin{equation}
\delta:S\times \Gamma \times S  \rightarrow \{0, 1/2, 1\}.
\end{equation}

Note: For any $s\in S$ and any $\sigma\in\Gamma$, $\delta(s,\sigma, t)$ is a so-called coin-tossing
distribution\footnote{A coin-tossing distribution on a finite set $Q$ is a mapping $\phi$ from
$Q$ to \{0, 1/2, 1\} such that $\sum_{q\in Q}\phi(q)=1$, which means choosing $q$ with probability $\phi(q)$.} on $S$ such that $\sum_{t\in S}\delta(s,\sigma,t)=1$.
For example, $\delta(s, \sigma, t)$ means that if $\mathcal{A}$ is in the state $s$ with the tape head scanning the symbol $\sigma$,
then the automaton enters the state $t$ with probability
         $\delta(s,\sigma,t)$.

\end{itemize}
\end{definition}

For an input string $\omega=\sigma_1\ldots\sigma_l$, the probability
distribution on the states of $\mathcal{A}$ during its acceptance process can be
traced using $n$-dimensional vectors. It is assumed that $\mathcal{A}$ starts to
process the input word written on the input tape as $w=\ |\hspace{-1.5mm}c\ \omega\$$
and let $v_0=(1,0,\ldots,0)^T_{n\times 1}$ denote the initial probability
distribution on states. If, during the acceptance process, the current
probability distribution vector is $v$ and a tape symbol $\sigma$ is
read, then the new state probability distribution vector will be, after
the automaton step, $u=A_{\sigma}v$, where $A_{\sigma}$ is such a matrix
that $A_{\sigma}(i,j)=\delta(s_j,\sigma,s_i)$. We then use $v_{|w|}=A_{\$}A_{\sigma_l}\cdots A_{\sigma_1}A_{\ |\hspace{-1mm}c }v_0$
to denote the final
probability distribution on states in case of the input $\omega$. The accepting probability of $\mathcal{A}$ with input $\omega$ is then
\begin{equation}
Pr[\mathcal{A}\  {\it accepts}\  \omega]=\sum_{s_i\in S_{acc}} v_{|w|}(i),
\end{equation}
where  $v_{|w|}(i)$ denotes the $i$th entry of $v_{|w|}$.

\begin{definition}
A {\it measurement-once one-way quantum
         automaton} (MO-1QFA) $\mathcal{A}$ is specified by a 5-tuple
\begin{equation}
\mathcal{A}=(Q,\Sigma,\Theta, |q_0\rangle,Q_{acc}),
\end{equation}
where:
\begin{itemize}
\item[1.] $Q$ is a finite set of quantum orthogonal states;

\item[2.] $\Sigma$ is a finite set of input symbols; $\Sigma$ is then extended to the tape symbol
set $\Gamma=\Sigma\cup\{\ |\hspace{-1.5mm}c,\$\}$, where $\ |\hspace{-1.5mm}c\notin \Sigma $ is called the left end-marker
and $\$\notin \Sigma$ is called the right end-marker;

\item[3.] $|q_{0}\rangle\in Q$ is the initial quantum state;

\item[4.] $Q_{acc}\subset Q$ is the set of accepting quantum states;

\item[5.] For each $\sigma\in \Gamma$, a unitary
 transformation $\Theta_{\sigma}$ is defined on the Hilbert space spanned by the states
  from $Q$.

\end{itemize}
\end{definition}

We describe the acceptance process of $\mathcal{A}$ for any given
input string $\omega=\sigma_1\cdots \sigma_l$ as follows. The automaton $\mathcal{A}$ states with the initial state $|q_0\rangle$,
reading the left-marker $\ |\hspace{-1.5mm}c$.  Afterwards, the unitary transformation $\Theta_{|\hspace{-1mm}c }$ is applied on $|q_0\rangle$.
After that, $\Theta_{|\hspace{-1mm}c }|q_0\rangle$ becomes the current state and the automaton reads $\sigma_1$.
The process continues until $\mathcal{A}$ reads $\$$
and ends in the state $|\psi_{\omega}\rangle=\Theta_{\$}\Theta_{\sigma_l}\cdots \Theta_{\sigma_1} \Theta_{|\hspace{-1mm}c }|q_0\rangle$.
Finally, a measurement is performed on $|\psi_{\omega}\rangle$ and the accepting probability of $\mathcal{A}$ on the input $\omega$ is equal to
\begin{equation}
Pr[\mathcal{A}\  {\it accepts}\  \omega]=\langle \psi_{\omega}|P_a|\psi_{\omega}\rangle=|| P_a|\psi_{\omega}\rangle ||^2,
\end{equation}
where $P_a=\sum_{q\in Q_{acc}}|q\rangle\langle q|$ is the projection onto the subspace spanned by $\{|q\rangle:|q\rangle\in Q_{acc}\}$.

\begin{definition}
A {\it measurement-many one-way quantum
         automaton} (MM-1QFA) $\mathcal{A}$
 is specified by a 6-tuple
\begin{equation}
\mathcal{A}=(Q,\Sigma,\Theta, |q_0\rangle,Q_{acc},Q_{rej}),
\end{equation}
where $Q$, $\Sigma$, $\Theta$, $|q_0\rangle$, $Q_{acc}$, and the tape symbol
set $\Gamma$ are the same as those defined above in an MO-1QFA. $Q_{rej}\subset Q$ is the set of rejecting states.
\end{definition}

For any given input string $\omega=\sigma_1\cdots \sigma_l$, the acceptance process is similar to that of MO-1QFA
except that after every transition, MM-1QFA $\mathcal{A}$ measures its state with respect to the three subspaces that are spanned
by the three subsets $Q_{acc}$, $Q_{rej}$ and $Q_{non}$, respectively, where $Q_{non}=Q\setminus(Q_{acc}\cup Q_{rej})$.
In other words, the projective measurement consists of $\{P_a, P_r, P_n\}$, where $P_a=\sum_{q\in Q_{acc}}|q\rangle\langle q|$,
$P_r=\sum_{q\in Q_{rej}}|q\rangle\langle q|$ and $P_n=\sum_{q\in Q_{non}}|q\rangle\langle q|$.
The accepting and rejecting probability are given as follows (for convenience, we denote $\sigma_0=\ |\hspace{-1.5mm}c$ and $\sigma_{l+1}=\$$):
\begin{equation}
Pr[\mathcal{A}\  {\it accepts}\  \omega]=\sum_{k=0}^{l+1}||P_a\Theta_{\sigma_k}\prod_{i=0}^{k-1}(P_n\Theta_{\sigma_i})|q_0\rangle ||^2,
\end{equation}
\begin{equation}
Pr[\mathcal{A}\ {\it reject}\  \omega]=\sum_{k=0}^{l+1}||P_r\Theta_{\sigma_k}\prod_{i=0}^{k-1}(P_n\Theta_{\sigma_i})|q_0\rangle ||^2.
\end{equation}
An important convention: In this paper we define $\prod_{i=1}^nA_i=A_nA_{n-1}\cdots A_1$, instead of the usual one $A_1A_2\cdots A_n$.

\begin{definition}
A {\it one-way quantum finite
automata with control language} (1QFACL) $\mathcal{A}$
 is specified by as a 6-tuple
\begin{equation}
\mathcal{A}=(Q,\Sigma,\Theta, |q_0\rangle, \mathcal{O}, \mathcal{L}),
\end{equation}
where:
\begin{itemize}
\item[1.] $Q$, $\Sigma$, $\Theta$, $|q_0\rangle$ and the tape symbol
set $\Gamma$ are the same as those defined above in an MO-1QFA;

\item[2.] $\mathcal{O}$ is an observable with the set of possible eigenvalues $\mathcal{C}=\{c_1, \cdots, c_s\}$ and
the projector set $\{P(c_i):i=1, \cdots, s\}$ where $P(c_i)$ denotes the projector onto the eigenspace corresponding to $c_i$;
\item[3.] $\mathcal{L}\subset \mathcal{C}^*$ is a regular language (called here as a control language).
\end{itemize}
\end{definition}

The input word $\omega=\sigma_1\cdots \sigma_l$ to 1QFACL $\mathcal{A}$ is in the form:
$w=\ |\hspace{-1.5mm}c\omega\$$ (for convenience, we denote $\sigma_0=|\hspace{-1.5mm}c$ and $\sigma_{l+1}=\$$).
Now, we define the behavior of $\mathcal{A}$ on the word $w$. The computation starts in the state $|q_0\rangle$, and then the
transformations associated with symbols in the word  $w$ are applied in succession. The transformation associated
with any symbol $\sigma\in\Gamma$ consists of two steps:

\begin{enumerate}
\item[1.] Firstly, $\Theta_{\sigma}$ is applied to the current state
$|\phi\rangle$ of $\mathcal{A}$, yielding the new state
$|\phi'\rangle=\Theta_{\sigma}|\phi\rangle$.

\item[2.] Secondly, the observable $\mathcal{O}$ is measured on
$|\phi'\rangle$. According to quantum mechanics principle, this
measurement yields result $c_k$ with probability
$p_k=||P(c_k)|\phi'\rangle||^2$, and the state of $\mathcal{A}$
collapses to $P(c_k)|\phi'\rangle/\sqrt{p_k}$.

\end{enumerate}

Thus, the computation on the word $w$ leads to a string
$y_0y_1\dots y_{l+1}\in \mathcal{C}^{*}$ with probability $p(y_0y_1\dots
y_{l+1}|\sigma_0\sigma_1\dots \sigma_{l+1})$ given by
\begin{equation}
p(y_0y_1\dots
y_{l+1}|\sigma_0\sigma_1\dots \sigma_{l+1})=||\prod^{l+1}_{i=0}(P(y_i)\Theta_{\sigma_i})|q_0\rangle||^2.
\end{equation}
A computation leading to a word $y\in \mathcal{C}^{*}$ is said to be accepted if $y\in \mathcal{L}$.
 Otherwise, it is rejected. Hence, the accepting probability of 1QFACL $\mathcal{A}$ is defined as:
\begin{equation}
Pr[\mathcal{A}\  {\it accepts}\  \omega]=\sum_{y_0y_1\dots y_{l+1}\in\mathcal{
L}}p(y_0y_1\dots y_{l+1}|\sigma_0\sigma_1\dots \sigma_{l+1})
\end{equation}

\subsection{Definition of 1QCFA}
In this subsection we introduce 1QCFA
and its acceptance process formally and in details.

2QCFA were first introduced by Ambainis and Watrous \cite{AJ}, and then studied by Qiu, Yakaryilmaz and etc. \cite{Qiu2,AY2,Zheng}.
1QCFA are the one-way version of 2QCFA.  Informally, we describe a 1QCFA as a DFA which has access to a quantum
 memory of a constant size (dimension), upon which it performs
quantum transformations and measurements. Given a finite set of quantum states $Q$, we denote by $\mathcal{H}(Q)$
the Hilbert space spanned by $Q$. Let
$\mathcal{U}(\mathcal{H}(Q))$ and $\mathcal{O}(\mathcal{H}(Q))$
denote the sets of unitary operators and projective measurements over $\mathcal{H}(Q)$, respectively.

\begin{definition}
A {\it one-way finite automata with quantum and classical states} (1QCFA) $\mathcal{A}$ is specified by a 10-tuple
\begin{equation}
\mathcal{A}=(Q,S,\Sigma,\Theta,\Delta,\delta,|q_{0}\rangle,s_{0},S_{acc},S_{rej})
\end{equation}

where:
\begin{itemize}
\item[1.] $Q$ is a finite set of quantum states;
\item[2.] $S$, $\Sigma$ and the tape symbol
set $\Gamma$ are the same as those defined above in a 1PFA;

\item[3.] $|q_{0}\rangle\in Q$ is the initial quantum state;

\item[4.] $s_{0}\in S$ is the initial classical state;

\item[5.] $S_{acc}\subset S$ and $S_{rej}\subset S$ are the sets of
classical accepting and rejecting states, respectively;

\item[6.] $\Theta$ is the mapping:
\begin{equation}
\Theta:S\times \Gamma \rightarrow
 \mathcal{U}(\mathcal{H}(Q)),
\end{equation}
assigning to each pair $(s,\gamma)$ a
 unitary transformation;

\item[7.] $\Delta$ is the mapping:
\begin{equation}
\Delta:S\times \Gamma \rightarrow
 \mathcal{O}(\mathcal{H}(Q)),
\end{equation}
where each $\Delta(s,\gamma)$ corresponds to a projective measurement
(a projective measurement will be taken each time a unitary
   transformation is applied; if we do not need a measurement,
we denote that $\Delta(s,\gamma)=I$, and we assume the result of the measurement to be $\varepsilon$ with certainty);

\item[8.] $\delta$ is a special transition function of classical states.
Let the results set of the measurement be $\mathcal{C}=\{c_{1},c_{2},\dots$,
$c_{s}\}$, then
\begin{equation}
\delta:S\times \Gamma \times \mathcal{C}\rightarrow
S,
\end{equation}
 where $\delta(s,\gamma)(c_{i})=s'$ means
that if a tape symbol $\gamma\in \Gamma$ is being
 scanned and the projective measurement result is $c_{i}$, then the  state $s$ is changed to $s'$.
\end{itemize}
\end{definition}

Given an input $\omega=\sigma_1\cdots\sigma_l$, the word on the tape will be
$w=|\hspace{-1.5mm}c\ \omega\$$ (for convenience, we denote $\sigma_0=|\hspace{-1.5mm}c$ and $\sigma_{l+1}=\$$).
Now, we define the behavior of 1QCFA $\mathcal{A}$ on the word $w$.
The computation starts in the classical state $s_0$ and the quantum state $|q_0\rangle$, then the
transformations associated with symbols in the word  $\sigma_0\sigma_1\cdots, \sigma_{l+1}$ are applied in succession.
The transformation associated
with a state $s\in S$ and a symbol $\sigma\in\Gamma$ consists of three steps:
\begin{enumerate}
\item[1.] Firstly, $\Theta(s,\sigma)$ is applied to the current quantum state
$|\phi\rangle$, yielding the new state
$|\phi'\rangle=\Theta(s,\sigma)|\phi\rangle$.

\item[2.] Secondly, the observable $\Delta(s,\sigma)=\mathcal{O}$ is measured on
$|\phi'\rangle$. The set of possible results is $\mathcal{C}=\{c_1, \cdots, c_s\}$.
According to such a quantum mechanics principle, such a
measurement yields the classical outcome $c_k$ with probability
$p_k=||P(c_k)|\phi'\rangle||^2$, and the quantum state of $\mathcal{A}$
collapses to $P(c_k)|\phi'\rangle/\sqrt{p_k}$.

\item[3.] Thirdly, the current classical state $s$ will be changed to $\delta(s,\sigma)(c_k) =s'.$

\end{enumerate}
An input word $\omega$ is assumed to be accepted (rejected) if and only if the classical state after scanning $\sigma_{l+1}$
is an accepting (rejecting) state. We assume that $\delta$ is well defined so that 1QCFA $\mathcal{A}$ always accepts or
rejects at the end of the computation.

Let $L\subset \Sigma^*$ and $0\leq\epsilon<1/2$, then 1QCFA $\mathcal{A}$ recognizes $L$ with bounded error $\epsilon$ if
\begin{itemize}
\item[1.] For any $\omega\in L$, $Pr[\mathcal{A}\ {\it accepts}\  \omega]\geq 1-\epsilon$, and
\item[2.] For any $\omega\notin L$, $Pr[\mathcal{A}\ {\it rejects}\  \omega]\geq 1-\epsilon$.
\end{itemize}

\section{Simulation of other models by 1QCFA}
In this section, we prove that the following automata models can be simulated by 1QCFA: DFA, 1PFA, MO-1QFA, MM-1QFA and 1QFACL.

\begin{theorem} \label{th1}
Any $n$ states DFA $\mathcal{A}=(S,\Sigma,\delta,s_0, S_{acc})$ can be simulated by a
1QCFA $\mathcal{A}'=(Q',S',\Sigma',\Theta',\Delta',\delta',|q_{0}\rangle',s_{0}',S_{acc}',S_{rej}')$
with $1$ quantum state and $n+1$ classical states.
\end{theorem}

\begin{proof}
Actually, if we do not use the quantum component of 1QCFA, the automaton is reduced to a DFA.
Let $Q'=\{|q_{0}\rangle'\}$, $S'=S\cup \{s_r\}$, $\Sigma'=\Sigma$, $s_{0}'=s_0$, $S_{acc}'=S_{acc}$ and $S_{rej}'=\{s_{r}\}$.
For any $s\in S$ and any $\sigma\in \Sigma$, let $\Theta(s,\sigma)=I$,  $\Delta'(s,\sigma)=I$,
and the classical transition function $\delta'$ is defined as follows:
\begin{equation}
  \delta'(s,\sigma)(c)=\left\{%
\begin{array}{ll}
    s,                 & \sigma=\ |\hspace{-1.5mm}c;\\
    \delta(s,\sigma),\ \ \ \ \ \   & \sigma\in \Sigma,\\
    s,                 & \sigma=\$, s\in S_{acc}';\\
    s_{r},            & \sigma=\$, s\notin S_{acc}'.\\

\end{array}%
\right.
\end {equation}
where $c$ is the measurement result.

\end{proof}

\begin{theorem}
Any $n$ states 1PFA
$\mathcal{A}^{1}=(S^{1},\Sigma^{1},\delta^{1},s_{1}^{1},S_{acc}^{1})$
can be simulated by a 1QCFA
$\mathcal{A}^{2}=(Q^{2},S^{2},\Sigma^{2},\Theta^{2},\Delta^{2},\delta^{2},|q_{0}\rangle^{2},s_{0}^{2},S_{acc}^{2},S_{rej}^{2})$
with $2$ quantum states and $n+1$ classical states.
\end{theorem}
\begin{proof}
A 1PFA is essentially a DFA augmented with a fair coin-flip component. In every transition, 1PFA can use a fair coin-flip or not freely.
Using the quantum component, a 1QCFA can simulate the fair coin-flip perfectly.
\begin{lemma}
A fair coin-flip can be simulate by 1QCFA $\mathcal{A}$ with two quantum states, a unitary operation and a projective measurement.
\end{lemma}
\begin{proof}
 The automaton
$\mathcal{A}$ simulates a coin-flip according to the following
transition functions, with $|p_0\rangle$ as the starting
 quantum state. We use two orthogonal basis states $|p_0\rangle$ and $|p_1\rangle$.  Let a projective measurement $M=\{P_0,P_1\}$ be defined by
\begin{equation}
 P_0=|p_0\rangle\langle p_0|, P_1=|p_1\rangle\langle p_1|.
\end {equation}
 The results 0 and 1 represent the results of coin-flip ``head" and
``tail", respectively. The corresponding unitary operation
   will be
\begin{equation}
  U=\left(%
\begin{array}{cc}
  \frac{1}{\sqrt{2}} &  \frac{1}{\sqrt{2}} \\
   \frac{1}{\sqrt{2}} &  -\frac{1}{\sqrt{2}} \\
\end{array}%
 \right).
\end {equation}
 This operator changes the state $|p_0\rangle$ or $|p_1\rangle$
 to a superposition state $|\psi\rangle$ or $|\phi\rangle$, respectively, as
 follows:
\begin{equation}
|\psi\rangle=\frac{1}{\sqrt{2}}(|p_0\rangle+|p_1\rangle),\ \  |\phi\rangle=\frac{1}{\sqrt{2}}(|p_0\rangle-|p_1\rangle).
\end{equation}
When measuring $|\psi\rangle$ or $|\phi\rangle$
 with $M$, we will get the result  0 or 1 with probability
 $\frac{1}{2}$, respectively. This is similar to a coin-flip
 process. If the result is 0, we simulate ``head" result of the coin-flip; if the result
 is 1, we simulate ``tail" result of the coin-flip. So the Lemma is proved.
\end{proof}

If the current state of 1PFA $\mathcal{A}^{1}$ is $s$ and the
scanning symbol is $\sigma\in \Sigma$, $\mathcal{A}^{1}$ makes a
coin-flip. The current state of $\mathcal{A}^{1}$ will change to
$t_1$ or $t_2$, in both cases with probability $\frac{1}{2}$.  We
use a 1QCFA $\mathcal{A}^{2}$ to simulate this step as follows:
\begin{enumerate}
  \item [1.] Use the quantum component of 1QCFA $\mathcal{A}^{2}$ to simulate a fair coin-flip. We assume the outcome to be $0$ or $1$.
  \item [2.] We define $\delta^{2}(s,\sigma)(0)=t_1$ and $\delta^{2}(s,\sigma)(1)=t_2$.
\end{enumerate}
The other parts of the simulation are similar to the one described in the proof of Theorem \ref{th1}.
\end{proof}

\begin{theorem}
Any $n$ quantum states MO-1QFA $\mathcal{A}^{1}=(Q^{1},\Sigma^{1},\Theta^{1}, \linebreak[0] |q_0\rangle^{1},\linebreak[0] Q_{acc}^{1})$ can be simulated by a 1QCFA $\mathcal{A}^{2}=(Q^{2},S^{2},\Sigma^{2},\Theta^{2},\Delta^{2},\delta^{2},|q_{0}\rangle^{2},s_{0}^{2},\linebreak[0]S_{acc}^{2},\linebreak[0]S_{rej}^{2})$ with $n$ quantum states and $3$ classical states.
\end{theorem}
\begin{proof}
We use the quantum component of 1QCFA to simulate the evolution of
quantum states of MO-1QFA and use the classical states of 1QCFA to
calculate the accepting probability. Let $Q^{2}=Q^{1}$,
$S^{2}=\{s_{0}^{2}, s_a^{2}, s_r^{2}\}$, $\Sigma^{2}=\Sigma^{1}$,
$|q_{0}\rangle^{2}=|q_{0}\rangle^{1}$, $S_{acc}^{2}=\{s_a^{2}\}$
and $S_{rej}^{2}=\{s_r^{2}\}$. For any current classical state $s$
and scanning symbol $\sigma$, the quantum transition function is
defined to be
\begin{equation}
 \Theta^{2}(s,\sigma)=\Theta^{1}(\sigma).
\end {equation}
The measurement function is defined to be
\begin{equation}
  \Delta^{2}(s,\sigma)=\left\{%
\begin{array}{ll}
    I,                 & \sigma\neq\$;\\
     \{P_a, P_r\},\ \ \                  & \sigma=\$.\\

\end{array}%
\right.
\end {equation}
where $P_a=\sum_{q\in Q_{acc}}|q\rangle\langle q|$, $P_r=I-P_a$. If we assume the outcome to be $c_a$ or $c_r$, then
the classical transition function will be defined to be
\begin{equation}
  \delta^{2}(s,\sigma)(c)=\left\{%
\begin{array}{ll}
    s, \ \ \ \ \ \ \ \ \                 & \sigma\neq\$;\\
    s_a^2,                & \sigma=\$, c=c_a;\\
    s_r^2,                & \sigma=\$, c=c_r.\\
\end{array}%
\right.
\end {equation}

\end{proof}

\begin{theorem}\label{th4}
Any $n$ quantum states MM-1QFA $\mathcal{A}^1=(Q^1,\Sigma^1,\Theta^1, \linebreak[0] |q_0\rangle^1,\linebreak[0] Q_{acc}^1,\linebreak[0] Q_{rej}^1)$
can be simulated by a 1QCFA $\mathcal{A}^{2}=(Q^{2},S^{2},\Sigma^{2},\Theta^{2},\Delta^{2},\delta^{2},\linebreak[0] |q_{0}\rangle^{2},\linebreak[0] s_{0}^{2},\linebreak[0] S_{acc}^{2},\linebreak[0] S_{rej}^{2})$ with $n$ quantum  states and $3$ classical states.
\end{theorem}

\begin{proof}
We use the quantum component of 1QCFA to simulate both the
evolution of quantum states of MM-1QFA and its projective
measurements. We use the classical states of 1QCFA to calculate
the accepting and rejecting probability. Let $Q^{2}=Q^{1}$,
$S^{2}=\{s_{0}^{2}, s_a^{2}, s_r^{2}\}$, $\Sigma^{2}=\Sigma^{1}$,
$|q_{0}\rangle^{2}=|q_{0}\rangle^{1}$, $S_{acc}^{2}=\{s_a^{2}\}$
and $S_{rej}^{2}=\{s_r^{2}\}$. For any current classical state $s$
and any scanning symbol $\sigma$, the quantum transition function
is defined to be
\begin{equation}
 \Theta^{2}(s,\sigma)=\Theta^{1}(\sigma).
\end {equation}
The measurement function is defined to be
\begin{equation}
\Delta^{2}(s,\sigma)=\{P_a, P_r, P_n\},
\end {equation}
where $P_a=\sum_{q\in Q_{acc}}|q\rangle\langle q|$, $P_r=\sum_{q\in Q_{rej}}|q\rangle\langle q|$ and $P_n=\sum_{q\in Q_{non}}|q\rangle\langle q|$. If we assume the classical outcomes to be $c_a$, $c_r$ or $c_n$, then the classical transition function will be defined to be

\begin{equation}
  \delta^{2}(s,\sigma)(c)=\left\{%
\begin{array}{ll}
    s_a^2,                 & s=s_a^2;\\
    s_r^2,                 & s=s_r^2;\\
    s_a^2,                 & s=s_0^2,c=c_a;\\
    s_r^2,                 & s=s_0^2,c=c_r;\\
    s_0^2, \ \ \ \ \ \ \ \ \ & s=s_0^2, c=c_n, \sigma\neq \$;\\
    s_r^2, \ \ \ \ \ \ \ \ \ & s=s_0^2, c=c_n, \sigma=\$.\\

\end{array}%
\right.
\end {equation}
\end{proof}

Although 1QFACL can accept all regular languages, their behavior
seems to be rather complicated. We prove that any 1QFACL can be
simulated by a 1QCFA with an easy to
 understand behavior.

\begin{theorem}
Any $n$ quantum states 1QFACL
$\mathcal{A}^1=(Q^1,\Sigma^1,\Theta^1, |q_0\rangle^1,\linebreak[0]
\mathcal{O}^1, \linebreak[0] \mathcal{L}^1)$, whose control
language $\mathcal{L}^1$ can be recognized by an $m$ states DFA
$\mathcal{A}=(S,\Sigma,\delta,\linebreak[0]s_0,\linebreak[0]
S_{acc})$, can be simulated by a 1QCFA
$\mathcal{A}^{2}=(Q^{2},S^{2},\Sigma^{2},\Theta^{2},\Delta^{2},\delta^{2},\linebreak[0]
|q_{0}\rangle^{2},\linebreak[0] s_{0}^{2},\linebreak[0]
S_{acc}^{2},\linebreak[0] S_{rej}^{2})$ with  $n$ quantum states
and $m+1$ classical states.
\end{theorem}
\begin{proof}
We use the quantum component of 1QCFA to simulate the evolution of
quantum states of 1QFACL and also its projective measurements. We
use the classical states of 1QCFA to simulate DFA $\mathcal{L}^1$.
Let $Q^{2}=Q^{1}$, $S^{2}=S\cup \{s_r\}$, $\Sigma^{2}=\Sigma^{1}$,
$s_0^2=s_0$, $|q_{0}\rangle^{2}=|q_{0}\rangle^{1}$,
$S_{acc}^{2}=S_{acc}$ and $S_{rej}^{2}=\{s_r\}$. For any current
classical state $s$ and any scanning symbol $\sigma$, the quantum
transition function will be defined to be
\begin{equation}
 \Theta^{2}(s,\sigma)=\Theta^{1}(\sigma).
\end {equation}
The measurement function is defined to be
\begin{equation}
  \Delta^{2}(s,\sigma)=\{P(c_i):i=1, \cdots, t\},
\end {equation}
where $P(c_i)$ denotes the projector onto the eigenspace corresponding to $c_i$. We assume that the set of possible
  classical outcomes is
 $\mathcal{C}=\{c_1, \cdots, c_t\}$, where $\mathcal{C}=\Sigma$, then the classical transition function will be  defined to be

\begin{equation}
  \delta^{2}(s,\sigma)(c)=\left\{%
\begin{array}{ll}
    \delta(s,c), \ \ \ \ \ \ \ \ \ & \sigma\neq \$;\\
    \delta(s,c),                   & \sigma=\$, \delta (s, c)\in S_{acc};\\
    s_r,                           & \sigma=\$, \delta (s, c)\notin S_{acc}.\\

\end{array}%
\right.
\end {equation}
\end{proof}

\section{Closure proprieties of
 languages accepted by
 1QCFA}
For convenience, we denote by 1QCFA($\epsilon$) the classes of
languages recognized by 1QCFA with bounded error $\epsilon$.
Moreover, let $QS(\mathcal{A})$ and $CS(\mathcal{A})$ denote the
numbers of quantum states and classical states of a 1QCFA
$\mathcal{A}$. We start to consider the operation of intersection
.

\begin{theorem}\label{thint}
If $L_1\in 1QCFA(\epsilon_1)$ and $L_2\in 1QCFA(\epsilon_2)$, then
$L_1\cap L_2\in 1QCFA (\epsilon)$, where
$\epsilon=\epsilon_1+\epsilon_2-\epsilon_1\epsilon_2$.
\end{theorem}

\begin{proof}
Let
$\mathcal{A}^i=(Q^i,S^i,\Sigma^i,\Theta^i,\Delta^i,\delta^i,|q_{0}\rangle^i,s_{0}^i,S_{acc}^i,S_{rej}^i)$
be 1QCFA to recognize $L_i$ with bounded error $\epsilon_i$
(i=1,2). We construct a 1QCFA
$\mathcal{A}=(Q,S,\Sigma,\Theta,\linebreak[0]\Delta,\linebreak[0]\delta,\linebreak[0]|q_{0}\rangle,s_{0},\linebreak[0]S_{acc},S_{rej})$
where:
\begin{enumerate}
  \item $Q=Q^1\otimes Q^2$,
  \item $S=S^1\times S^2$,
  \item $\Sigma=\Sigma^1\cap \Sigma^2$,
  \item $s_{0}=\langle s_0^1, s_0^2\rangle $,
  \item $|q_{0}\rangle=|q_{0}\rangle^1\otimes |q_{0}\rangle^2$,
  \item $S_{acc}=S_{acc}^1\times S_{acc}^2$,
  \item $S_{rej}=(S_{acc}^1\times S_{rej}^2)\cup(S_{rej}^1\times S_{acc}^2) \cup (S_{rej}^1\times S_{rej}^2)$

\item For any classical state $s=\langle s^1,s^2\rangle\in S$ and
any $\sigma\in \Sigma$, the quantum transition function of
$\mathcal{A}$ is defined to be
  \begin{equation}
  \Theta(s,\sigma)=\Theta(\langle s^1,s^2\rangle,\sigma)=\Theta^1(s^1,\sigma)\otimes \Theta^2(s^2,\sigma).
  \end{equation}

\item For any classical state $s=\langle s^1,s^2\rangle\in S$ and
any $\sigma\in \Sigma$, the measurement function of $\mathcal{A}$
is defined to be
  \begin{equation}
  \Delta(s,\sigma)=\Delta(\langle s^1,s^2\rangle,\sigma)=\Delta^1(s^1,\sigma)\otimes \Delta^2(s^2,\sigma).
  \end{equation}
As classical measurements outcomes are
 then tuples
 $c_{ij}=\langle c_i,c_j\rangle$.

\item For any classical state $s=\langle s^1,s^2\rangle\in S$ and
any $\sigma\in \Sigma$, the classical transition function of
$\mathcal{A}$ is defined to be
  \begin{equation}
\delta(s,\sigma)(c_{ij})=\delta(\langle
s^1,s^2\rangle,\sigma)(\langle
c_i,c_j\rangle)=\langle\delta^1(s^1,\sigma)(c_i),
\delta^2(s^2,\sigma)(c_j)\rangle.
  \end{equation}
\end{enumerate}
In terms of the 1QCFA $\mathcal{A}$ constructed above, for any $\omega\in \Sigma^{*}$, we have:
\begin{enumerate}
\item If $\omega\in L_1\cap L_2$, then $\mathcal{A}$ will enter a
state $\langle t_1, t_2\rangle\in S_{acc}^1\times S_{acc}^2$ at
the end of the computation with probability at least
$(1-\epsilon_1)(1-\epsilon_2)$. $\mathcal{A}$ accepts $\omega$
with the probability at least
$(1-\epsilon_1)(1-\epsilon_2)=1-(\epsilon_1+\epsilon_2-\epsilon_1\epsilon_2)$.
\item If $\omega\in L_1$ but $\omega\notin L_2$, then
$\mathcal{A}$ will enter a state $\langle t_1, t_2\rangle\in
S_{acc}^1\times S_{rej}^2$ at the end of the computation with
probability at least $(1-\epsilon_1)(1-\epsilon_2)$. $\mathcal{A}$
rejects $\omega$ with the probability at least
$1-(\epsilon_1+\epsilon_2-\epsilon_1\epsilon_2)$. \item The case
$\omega\notin L_1$ but $\omega\in L_2$ is symmetric to the
previous one and therefore the same is the outcome. \item If
$\omega\notin L_1$ and $\omega\notin L_2$, then $\mathcal{A}$ will
enter a state $\langle t_1, t_2\rangle\in S_{rej}^1\times
S_{rej}^2$ at the end of the computation with probability at least
$(1-\epsilon_1)(1-\epsilon_2)$. $\mathcal{A}$ rejects $\omega$
with the probability at least
$1-(\epsilon_1+\epsilon_2-\epsilon_1\epsilon_2)$.
\end{enumerate}
So $L_1\cap L_2\in 1QCFA (\epsilon)$.
\end{proof}

\begin{remark}
According to the construction given above, let
$QS(\mathcal{A}^1)=n_1$,
$CS(\mathcal{A}^1)\linebreak[0]=\linebreak[0]m_1$,
$QS(\mathcal{A}^2)=n_2$ and $CS(\mathcal{A}^2)=m_2$, then
$QS(\mathcal{A})=n_{1}n_2$, $CS(\mathcal{A})=m_{1}m_2$.
\end{remark}

A similar outcome holds for the union operation.
\begin{theorem}
If $L_1\in 1QCFA(\epsilon_1)$ and $L_2\in 1QCFA(\epsilon_2)$, then
$L_1\cup L_2\in 1QCFA (\epsilon)$, where
$\epsilon=\epsilon_1+\epsilon_2-\epsilon_1\epsilon_2$.
\end{theorem}
\begin{proof}
Let
$\mathcal{A}^i=(Q^i,S^i,\Sigma^i,\Theta^i,\Delta^i,\delta^i,|q_{0}\rangle^i,s_{0}^i,S_{acc}^i,S_{rej}^i)$
be 1QCFA to recognize $L_i$ with bounded error $\epsilon_i$
(i=1,2). The construction of the 1QCFA
$\mathcal{A}=(Q,S,\Sigma,\Theta,\linebreak[0]\Delta,\linebreak[0]\delta,\linebreak[0]|q_{0}\rangle,s_{0},\linebreak[0]S_{acc},S_{rej})$
is the same as in the proof of Theorem \ref{thint} except for
$S_{acc}$ and $S_{rej}$. We define $S_{acc}=(S_{acc}^1\times
S_{rej}^2)\cup(S_{rej}^1\times S_{acc}^2) \cup (S_{acc}^1\times
S_{acc}^2) $ and $S_{rej}=S_{rej}^1\times S_{rej}^2$. The rest of
the proof is similar to the proof in Theorem \ref{thint}.
\end{proof}

\begin{remark}
In the last proof the set of input symbols was defined as
 $\Sigma=\Sigma^1\cap \Sigma^2$. Actually, if we take
$\Sigma=\Sigma^1\cup \Sigma^2$, the theorem still holds. In that
case, we extend $\Sigma^i$ to $\Sigma$ by adding a rejecting
classical state $s^i_r$ to $\mathcal{A}^i$. For any classical
state $s^i\in S^i$ and $\sigma^i\notin \Sigma^i$, the quantum
transition function is defined to be $\Theta^i(s^i,\sigma^i)=I$,
the measurement function is defined to be
$\Delta^i(s^i,\sigma^i)=I$. We assume the measurement result to be
$c$, then the classical transition function will be defined to be
$\delta^i(s^i,\sigma^i)(c)=s^i_r$. For the new adding state
$s^i_r$, we define the transition functions as follow: for any
$\sigma\in\Sigma$, $\Theta^i(s^i_r,\sigma)=I$,
$\Delta^i(s^i_r,\sigma)=I$, $\delta^i(s^i_r,\sigma)(c)=s^i_r$,
where $c$ is the the measurement result.
\end{remark}

\begin{theorem}
If $L\in 1QCFA(\epsilon)$, then also $L^{c}\in 1QCFA(\epsilon)$, where $L^{c}$ is the complement of $L$.
\end{theorem}
\begin{proof}
Let a 1QCFA($\epsilon$)
$\mathcal{A}=(Q,S,\Sigma,\Theta,\Delta,\delta,|q_{0}\rangle,s_{0},S_{acc},S_{rej})$
accept $L$ with a bounded error $\epsilon$. We can construct the
1QCFA $\mathcal{A}^{c}$ only by exchanging the classical accepting
and rejecting states in $\mathcal{A}$. That is,
$\mathcal{A}^{c}=(Q,S,\Sigma,\Theta,\Delta,\delta,|q_{0}\rangle,s_{0},\linebreak[0]
S_{acc}^{c},\linebreak[0] S_{rej}^{c})$, where
$S_{acc}^{c}=S_{rej}$, $S_{rej}^{c}=S_{acc}$ and the other
components remain the same as those  defined in $\mathcal{A}$.
Afterwards we have:
\begin{enumerate}
\item If $\omega\in L^{c}$,  then $\omega\notin L$. Indeed, for an
input $\omega$, $\mathcal{A}$ will enter a rejecting state with
probability at least $1-\epsilon$ at the end of the computation.
With the same input $\omega$, $\mathcal{A}^{c}$ will enter an
accepting state with probability at least  $1-\epsilon$ at the end
of the computation. Hence, $\mathcal{A}^{c}$ accepts $\omega$ with
the probability at least $1-\epsilon$; \item The case
$\omega\notin L^{c}$ is treated in a symmetric way..
\end{enumerate}
\end{proof}
\begin{remark}
According to the construction given above, if $QS(\mathcal{A})=n$,
$CS(\mathcal{A})=m$, then
$QS(\mathcal{A}^c)\linebreak[0]=\linebreak[0]n$,
$CS(\mathcal{A}^c)=m$.
\end{remark}

\section{Succinctness results}

State complexity and succinctness results are an important
research area of classical automata theory, see \cite{Yu}, with a
variety of
 applications. Once quantum versions of classical automata were introduced  and explored, it started to be of large interest to find out through
 succinctness results a relation between the power of classical and
 quantum automata model. This has turned out to be an area of surprising
 outcomes that again indicated that relations between
 classical and corresponding quantum automata models is intriguing. For
 example, it has been shown, see \cite{Amb-F,AN,ANT,Le}, that for some languages 1QFA
 require exponentially less states that classical 1FA, but for some other
 languages it can be in an opposite way.

Since 1QCFA can simulate both 1FA and 1QFA, and in this way they combine
the advantages of both of these models, it is of interest to explore the
relation between the state complexity of languages for the case that they
are accepted by 1QCFA and MM-1QFA and this we will do in this section.

The main result we obtain when considering languages
 $L_{m}=\{ a^*b^*\mid\  |a^*b^*|=km, k=1,2,\cdots\}$, where $m$ is a prime.
%$L_{m}$ is studied in \cite{QM}, and similar result has been given using {\it One-way quantum
%finite automata together with classical states} (1QFAC).
Obviously, there exist a $2m+2$ states DFA, depicted in Figure \ref{f1} that accepts $L_{m}$.
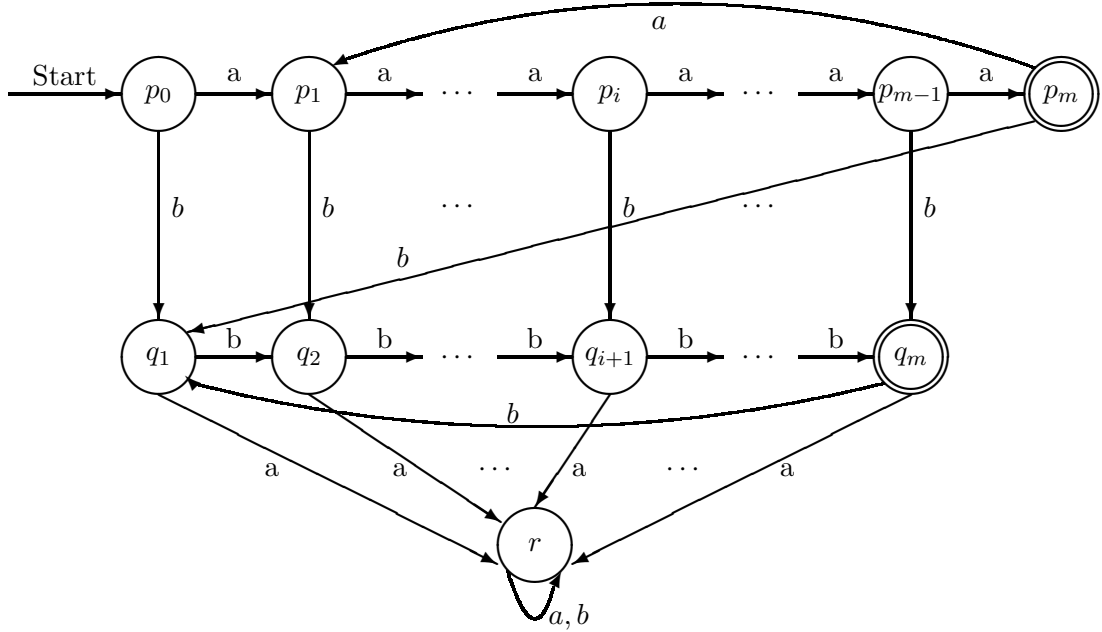
\begin{figure}
 %  %Requires \usepackage{graphicx}
  \centering
  \setlength{\unitlength}{1cm}
\begin{picture}(16,8)\thicklines
  \put(2,7){\circle{1}\makebox(0,0){$p_{0}$}}
\put(4,7){\circle{1}\makebox(0,0){$p_{1}$}}
\put(6,7){\makebox(0,0){$\cdots$}}
\put(8,7){\circle{1}\makebox(0,0){$p_{i}$}}
\put(10,7){\makebox(0,0){$\cdots$}}
\put(12,7){\circle{1}\makebox(0,0){$p_{m-1}$}}
\put(14,7){\circle{1}\circle{0.9}\makebox(0,0){$p_{m}$}}

\put(2,3.5){\circle{1}\makebox(0,0){$q_{1}$}}
\put(4,3.5){\circle{1}\makebox(0,0){$q_{2}$}}
\put(6,3.5){\makebox(0,0){$\cdots$}}
\put(8,3.5){\circle{1}\makebox(0,0){$q_{i+1}$}}
\put(10,3.5){\makebox(0,0){$\cdots$}}
\put(12,3.5){\circle{1}\circle{0.9}\makebox(0,0){$q_{m}$}}

\put(7,1){\circle{1}\makebox(0,0){$r$}}

\put(0,7){\vector(1,0){1.5}\makebox(-1.5,0.5){Start}}
\put(2.5,7){\vector(1,0){1}\makebox(-1,0.5){a}}
\put(4.5,7){\vector(1,0){1}\makebox(-1,0.5){a}}
\put(6.5,7){\vector(1,0){1}\makebox(-1,0.5){a}}
\put(8.5,7){\vector(1,0){1}\makebox(-1,0.5){a}}
\put(10.5,7){\vector(1,0){1}\makebox(-1,0.5){a}}
\put(12.5,7){\vector(1,0){1}\makebox(-1,0.5){a}}

\put(2.5,3.5){\vector(1,0){1}\makebox(-1,0.5){b}}
\put(4.5,3.5){\vector(1,0){1}\makebox(-1,0.5){b}}
\put(6.5,3.5){\vector(1,0){1}\makebox(-1,0.5){b}}
\put(8.5,3.5){\vector(1,0){1}\makebox(-1,0.5){b}}
\put(10.5,3.5){\vector(1,0){1}\makebox(-1,0.5){b}}

\put(2,6.5){\vector(0,-1){2.5}\makebox(0.5,-2){$b$}}
\put(4,6.5){\vector(0,-1){2.5}\makebox(0.5,-2){$b$}}
\put(6,6.5){\makebox(0,-2){$\cdots$}}
\put(8,6.5){\vector(0,-1){2.5}\makebox(0.5,-2){$b$}}
\put(10,6.5){\makebox(0,-2){$\cdots$}}
\put(12,6.5){\vector(0,-1){2.5}\makebox(0.5,-2){$b$}}

\put(13.64, 6.64){\vector(-4,-1){11.25}\makebox(5,-3.6){$b$}}

\qbezier(13.64,7.35) (9,9) (4.45,7.45)
\put(4.45,7.45){\vector(-2,-1){0.15}\makebox(8,1){$a$}}

\qbezier(11.64,3.15) (7,2) (2.45,3.15)
\put(2.45,3.15){\vector(-1,1){0.1}\makebox(8,-0.8){$b$}}

\put(2,3){\vector(2,-1){4.5}\makebox(-6,-2){a}}
\put(4,3){\vector(3,-2){2.55}\makebox(-2.7,-2){a}}
\put(6,3){\makebox(1,-2){$\cdots$}}
\put(8,3){\vector(-2,-3){1}\makebox(0.7,-2){a}}
\put(10,3){\makebox(-2,-2){$\cdots$}}
\put(12,3){\vector(-2,-1){4.5}\makebox(5,-2){a}}

\qbezier(6.64,0.64) (7,-0.5) (7.25,0.44)
\put(7.25,0.44){\vector(1,2){0.1}\makebox(0.2,-0.8){$a,b$}}
\end{picture}

  \centering\caption{DFA $\mathcal{A}$ recognizing $L_{m}$}\label{f1}
\end{figure}

\begin{lemma}
DFA $\mathcal{A}$ depicted in Figure \ref{f1} is minimal.
\end{lemma}
\begin{proof}
We show that any two different state $s$ and $t$ are
distinguishable (i.e., there exists a string $z$ such that exactly
one of the following states $\widehat{\delta}(p,z)$\footnote{For
any string $x\in \Sigma^*$ and any $\sigma\in \Sigma$,
$\widehat{\delta}(s,\sigma
x)=\widehat{\delta}(\delta(s,\sigma),x)$; if $|x|=0$,
$\widehat{\delta}(s,x)=s$ \cite{Hop}.} or $\widehat{\delta}(q, z)$
is an accepting state \cite{Yu}).
\begin{enumerate}
\item For $0\leq i\leq m$, $0\leq j\leq m$ and $i\neq j$, we have
$\widehat{\delta}(p_i, a^{m-i})=p_m$ and $\widehat{\delta}(p_j,
a^{m-i})=p_k$, where $k\neq m$. Hence, $p_i$ and $p_j$ are
distinguishable.

\item For $1\leq i\leq m$, $1\leq j\leq m$ and $i\neq j$, we have
$\widehat{\delta}(q_i, b^{m-i})=q_m$ and $\widehat{\delta}(p_j,
b^{m-i})=q_k$, where $k\neq m$. Hence, $q_i$ and $q_j$ are
distinguishable.

\item For $0\leq i\leq m$ and $1\leq j\leq m$, we have
$\widehat{\delta}(p_i, a^{m-i})=p_m$ and $\widehat{\delta}(q_j,
a^{m-i})=r$. Hence, $p_i$ and $q_j$ are distinguishable.

\item Obviously, the state $r$ is distinguishable from any other
state $s$.
\end{enumerate}
Therefore, the Lemma has been proved.
\end{proof}

\begin{lemma}[\cite{Amb-F}]\label{llp}
Any 1PFA recognizing $L_{m}$ with probability $1/2+\epsilon$, for a fixed $\epsilon>0$, has at least $m$ states.
\end{lemma}

\begin{remark}
The proof can be obtained by an easy modification of the proof from the
paper \cite{Amb-F} where the state complexity of the language $L_p=\{a^i\,|\,
 i\ \mbox{is divisible by}\ p\}$ is considered.
\end{remark}

\begin{lemma}[\cite{Amb-F}]\label{l3}(Forbidden construction)
Let $L$ be a regular language, and let $\mathcal{A}$ be its
minimal DFA. Assume that there is a word $w$ such that
$\mathcal{A}$ contains states $s, t$ (a forbidden construction)
satisfying:
\begin{enumerate}
  \item $s\neq t$,
  \item $\widehat{\delta}(s,x)=t$,
  \item $\widehat{\delta}(t,x)=t$ and
  \item $t$ is neither ``all-accepting" state, nor ``all-rejecting" state.
\end{enumerate}
Then $L$ cannot be recognized by an MM-1QFA with bounded error $\frac{7}{9}+\epsilon$ for any fixed $\epsilon>0$.
\end{lemma}

\begin{theorem}
For any fixed $\epsilon>0$, $L_{m}$ cannot be recognized by an MM-1QFA with bounded error $\frac{7}{9}+\epsilon$.
\end{theorem}
\begin{proof}
According to Lemma \ref{l3}, we know that $L_m$ cannot be accepted
by any MM-1QFA with bounded error $\frac{7}{9}+\epsilon$ since its
minimal DFA (see Figure \ref{f1}) contains the ``Forbidden
construction" of Lemma \ref{l3}. For example, we can take $s =
p_0$, $t = p_m$, $x = a^m$, then we have $\widehat{\delta}(p_0,
a^m)=p_m$, $\widehat{\delta}(p_m, a^m)=p_m$,
$\widehat{\delta}(p_m, b^m)=q_m$ and $\widehat{\delta}(p_m,
ba)=r$.
\end{proof}

Let $L_1=\{a^*b^* \}$ and $L_{2}=\{w \mid w\in\{a,b\}^*, |w|=km, k=1,2,\cdots \}$ where $m$ is a prime.
So we have $L_m=L_1\cap L_2$. We will show $L_1$ and $L_{2}$ can be recognized by 1QCFA.

\begin{lemma}\label{l4}
The language $L_1$ can be recognized by a 1QCFA $\mathcal{A}^1$ with certainty with 1 quantum state and 4 classical states.
\end{lemma}
\begin{proof}
$L_1$ can be accepted by a DFA $\mathcal{A}$ with 3 classical
states (see Figure \ref{f3}). According to Theorem \ref{th1},
$\mathcal{A}$ can be simulated by a 1QCFA $\mathcal{A}^1$ with 1
quantum state and 4 classical states.

\begin{figure}
 % %Requires \usepackage{graphicx}
\centering

\setlength{\unitlength}{1cm}
\begin{picture}(16,4)\thicklines
\put(6,3){\circle{1}\circle{0.9}\makebox(0,0){$p_{0}$}}
\put(10,3){\circle{1}\circle{0.9}\makebox(0,0){$p_{1}$}}
\put(6.5,3){\vector(1,0){3}\makebox(-3,0.5){b}}
\put(8,1){\circle{1}\makebox(0,0){$r$}}

\put(4,3){\vector(1,0){1.5}\makebox(-1.5,0.5){Start}}

\qbezier(5.64,3.35) (6,5) (6.35,3.45)
\put(6.35,3.45){\vector(1,-2){0.1}\makebox(0.2,1){$a$}}

\qbezier(9.64,3.35) (10,5) (10.35,3.45)
\put(10.35,3.45){\vector(1,-2){0.1}\makebox(0.2,1){$b$}}

\put(9.64,2.64){\vector(-1,-1){1.25}\makebox(0.5,-1.8){$a$}}

\put(8,1){\circle{1}\makebox(0,0){$r$}}

\qbezier(7.64,0.64) (8,-1) (8.25,0.44)
\put(8.25,0.44){\vector(1,2){0.1}\makebox(0.5,-0.8){$a,b$}}

\end{picture}
\centering\caption{A DFA recognizing the language $L_1$}\label{f3}
\end{figure}
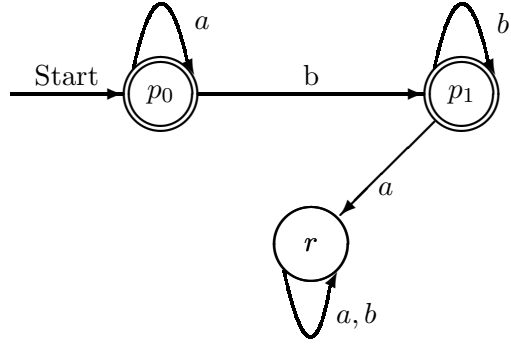

\end{proof}
\begin{lemma}[\cite{Amb-F}]\label{l5}
For any $\epsilon>0$, there is an MM-1QFA $\mathcal{A}$ with $\bf{O}(\log{m})$ quantum states recognizing $L_2$ with a bounded error $\epsilon$.
\end{lemma}
\begin{lemma}\label{l6}
For any $\epsilon>0$, there is a 1QCFA $\mathcal{A}^2$ with
$\bf{O}(\log{m})$ quantum states and 3 classical states
recognizing $L_2$ with a bounded error $\epsilon$.
\end{lemma}
\begin{proof}
According to Lemma \ref{l5}, there is an MM-1QFA $\mathcal{A}$
with $\bf{O}(\log{m})$ quantum states recognizing $L_2$ with
bounded error $\epsilon$. According to Theorem \ref{th4}, an
$\bf{O}(\log{m})$ quantum states  MM-1QFA $\mathcal{A}$ can be
simulated by a 1QCFA  with $\bf{O}(\log{m})$ quantum states and 3
classical states.
\end{proof}

\begin{theorem}
For any $\epsilon>0$, $L_{m}$ can be recognized by a 1QCFA with
$\bf{O}(\log{m})$ quantum states and 12 classical states with a
bounded error $\epsilon$.
\end{theorem}
\begin{proof}
$L_m=L_1\cap L_2$. According to Lemma \ref{l4}, the language $L_1$
can be recognized by 1QCFA $\mathcal{A}^1$ with 1 quantum state
and 4 classical states with certainty (i.e., $\epsilon_1=0$).
According to Lemma \ref{l6}, for any $\epsilon>0$, the language
$L_2$ can be recognized by 1QCFA $\mathcal{A}^2$ with
$\bf{O}(\log{m})$ quantum states and 3 classical states with a
bounded error $\epsilon$. According to Theorem \ref{thint}, 1QCFA
is closed under intersection. Hence, there is a 1QCFA
$\mathcal{A}$ recognize $L_m$ with a bounded error $\epsilon$.
Therefore $QS(\mathcal{A}^1)=1$, $CS(\mathcal{A}^1)=4$,
$QS(\mathcal{A}^2)=\bf{O}(\log{m})$ and $CS(\mathcal{A}^2)=3$, so
$QS(\mathcal{A})=QS(\mathcal{A}^1)\times
QS(\mathcal{A}^2)=\bf{O}(\log{m})$, $CS(\mathcal{A})=
CS(\mathcal{A}^1)\times CS(\mathcal{A}^2)=12$.
\end{proof}

\section{Conclusions}
2QCFA were introduced by Ambainis and Watrous \cite{AJ}. In this
paper, we investigated the one-way version of 2QCFA, namely 1QCFA.
Firstly, we gave a formal definition of 1QCFA. Secondly, we
showed that DFA, 1PFA, MO-1QFA, MM-1QFA and 1QFACL can be
simulated by 1QCFA.  As we know, the behavior of 1QFACL seems to
be rather complicated. However, when we used a 1QCFA to simulate a
1QFACL, the behavior of 1QCFA started to be seen as quite natural.
Thirdly, we studied closure properties of languages accepted by
 1QCFA, and we proved that the family of languages accepted by
1QCFA is closed under intersection, union, and complement.
Fourthly, for any fixed $\epsilon_1>0$ and any prime $m$ we have showed that the language $L_{m}=\{ a^*b^*\mid\  |a^*b^*|=km,
k=1,2,\cdots\}$, cannot be recognized
by any MM-1QFA with bounded error $\frac{7}{9}+\epsilon_1$, and
any 1PFA recognizing it has at last $m$ states, but $L_{m}$ can be
recognized by a 1QCFA for any error bound $\epsilon>0$ with
$\bf{O}(\log{m})$ quantum states and 12 classical states. Thus,
1QCFA can make use of merits of both 1FA and 1QFA.

To conclude, we would like to propose some problems for further consideration.
\begin{enumerate}
  \item Obviously, all regular languages can be recognized by 1QCFA. Is there any non-regular language recognized by 1QCFA?
  \item Are 1QCFA closed under catenation and reversal?
\end{enumerate}

\end{document}